\documentclass[letterpaper, 10pt, conference]{ieeeconf}      

\IEEEoverridecommandlockouts                              
\usepackage{amsmath,graphicx,amsfonts,amssymb,epsfig,subfigure,mathrsfs,mathtools}
\usepackage{color}
\usepackage{multirow}
\usepackage{multicol}
\usepackage{lipsum}
\usepackage{rotating}
\usepackage{graphicx}%
\usepackage{algorithm}
\usepackage{algpseudocode}
\usepackage{cite}
\usepackage{tikz}

\newcommand{\tr}{\text{ tr}}

\newtheorem{theorem}{Theorem}[section]
\newtheorem{corollary}{Corollary}[section]
\newtheorem{definition}{Definition}[section]
\newtheorem{lemma}{Lemma}[section]
\newtheorem{proposition}{Proposition}[section]
\newtheorem{remark}{Remark}[section]

\title{\LARGE \bf
Probabilistic Robustness Analysis of Stochastic Jump Linear Systems
}

\author{Kooktae Lee, Abhishek Halder, and Raktim Bhattacharya
\thanks{Kooktae Lee, Abhishek Halder, and Raktim Bhattacharya are with the Department of Aerospace Engineering, Texas A\&M University,
        College Station, TX 77843-3141, USA, {\tt\scriptsize \{animodor,ahalder,raktim\}@tamu.edu.}}%
}

\begin{document}
\maketitle
\thispagestyle{empty}
\pagestyle{empty}

\begin{abstract}
In this paper, we propose a new method to measure the probabilistic robustness of stochastic jump linear system with respect to both the initial state uncertainties and the randomness in switching. Wasserstein distance which defines a metric on the manifold of probability density functions is used as tool for the performance and the stability measures. Starting with Gaussian distribution to represent the initial state uncertainties, the probability density function of the system state evolves into mixture of Gaussian, where the number of Gaussian components grows exponentially. To cope with computational complexity caused by mixture of Gaussian, we prove that there exists an alternative probability density function that preserves exact information in the Wasserstein level. The usefulness and the efficiency of the proposed methods are demonstrated by example.
\end{abstract}

\section{Introduction}
Over few decades, broad investigations have been achieved for jump systems. A jump system is defined as a switching between a family of subsystem dynamics. Specifically, a jump linear system for which all subsystem dynamics are linear is subcategory of switched systems. According to a switching logic that orchestrates jump between different sublinear dynamics, jump linear systems further branch out into two subfields; deterministic and stochastic jump linear systems. 

Deterministic jump linear systems have been applied to power systems, manufacturing systems, aerospace systems, etc(\cite{liberzon2003switching,lin2009stability,cassandras2001optimal,minto1991new}).
In this system, a deterministic switching law governs the jump sequences which are constituted by deterministic process. The advantage of utilization for such deterministic jump systems stems from not only plant stabilization but also system performance, adaptive control, and resource-constrained scheduling. Since the stability is one of the most important issues on jump systems, a variety of results with respect to the stability are established and recent literature can be found in \cite{lin2009stability}. According to \cite{lin2009stability}, the sufficient condition for the stability of deterministic jump linear systems can be guaranteed by solving certain linear matrix inequalities(LMIs). Also, necessary and sufficient conditions for the stability are shown via a finite tuple, satisfying a certain condition.

On the other hands, stochastic jump linear systems where the switching law is fractional numbers representing jump probability are commonly used, such as abrupt environmental disturbances, component failures, changes in subsystems interconnections, abrupt changes in the operation point, random communication delays in control, etc. 
Unlike the deterministic jump linear systems, stochastic jump linear systems have several difficulties on the analysis. For instance, due to the randomness of switching laws we get different state trajectories for multiple run even with same initial conditions and switching probabilities. Moreover, the states of the stochastic jump linear system become random variables with a probability density function even if deterministic states are given initially. The spatio-temporal evolution of the joint probability density function for system states depends on the stochastic switching laws. Therefore, the analysis of the stability for stochastic jump linear systems has been investigated in terms of mean square sense with several different notions(asymptotically / exponentially / stochastically stable)\cite{feng1992stochastic}. A Markov jump linear system\cite{do2005discrete} where the switching rule is governed by Markovian process is one example of such stochastic jump linear systems. Because of the easiness to formulate the randomness in communication delay or packet loss in networked systems, Markov process describing the stochastic jump is widely adopted(\cite{xu2004optimal,you2011minimum,do2005discrete}). Moreover, further analysis has been accomplished even for an incomplete or partially known transition probability in Markovian process(\cite{zhang2008analysis,zhang2010necessary,xiong2005robust}). 

In the current paper we propose a probabilistic robustness analysis tool for any stochastic jump linear systems, but not necessarily for Markovian jump. The term robustness used here originate in the context of the ability for the system to resist given initial state uncertainties. Generally, initial states include uncertainties caused by measurement noise or sensor inaccuracy. Hence, we aim to analyse the robustness of stochastic jump linear systems with given initial state uncertainties. To assess both the stability and the robustness of the system, we adopt Wasserstein distance that defines a metric on the manifold of probability density functions. Under the assumption that the initial state uncertainties are given as Gaussian distribution, the state probability density function evolves into the combination of multi-Gaussian, which is the mixture of Gaussian(MoG). The stochastic switching laws give rise to the exponential growth of Gaussian components in the MoG, which incurs the computational complexity. In this paper, we prove that there exists an alternative way to analyse the robustness of stochastic jump linear systems using Wasserstein distance. This solution provides an identical information without any approximation errors while avoiding an MoG expression.


Rest of this paper is organized as follows. Section II introduces preliminaries for stochastic jump linear systems. Brief explanations of Wasserstein distance are described in Section III. Section IV shows how to propagate initial state uncertainties and compute Wasserstein distance. Then, section V demonstrates the usefulness of the proposed method by showing example. Finally, Section VI concludes this paper.

\section{Preliminaries}
\textbf{Notation:} Most notations are standard. The set of real and non-negative integer are denoted by $\mathbb{R}$ and $\mathbb{Z}^+$, respectively. Further, $\mathcal{I}$ represents the set of switching modes. We denote trace of a matrix using a notation$\tr\left(\cdot\right)$. Abbreviation m.s. stands for the convergence in mean-square sense. The notation $X \sim \varsigma\left(x\right)$ denotes that the random variable $X$ has probability density function (PDF) $\varsigma\left(x\right)$. The symbol $\mathcal{N}\left(\mu,\Sigma\right)$ is used to denote the PDF of a Gaussian random vector with mean $\mu$ and covariance $\Sigma$.

Consider a discrete-time jump linear system, given by 
\begin{eqnarray}
x(k+1) = A_{\sigma_k}x(k),\quad k\in\mathbb{Z}^{+}, \sigma_k\in\mathcal{I},\label{3}
\end{eqnarray}
where the state vectors $x\in \mathbb{R}^{n}$, the system matrices $A_{\sigma_{k}}\in\mathbb{R}^{n\times n}$ and the set of modes $\mathcal{I}=\{1,2,\cdots,m\}$. 

For stochastic jump linear systems, we can define a switching probability at time $k$ as $\pi(k)=[\pi_1(k),\pi_2(k),\cdots,\pi_m(k)]$, where $\displaystyle\sum_{j=1}^{m}\pi_j(k)=1$.

\begin{definition}(Stochastic jump linear system)
Tuples of the form $(\pi(k),A_{\sigma_k},\mathcal{I})$ is termed as discrete-time stochastic jump linear system, provided the mode dynamics are given by \eqref{3}. $\pi(k)$ denotes the time-varying occupation probability vectors for prescribed stochastic switching process $\sigma_k$.
\end{definition}


Stochastic jump linear systems have each element $\pi_j$, which is any fractional number, satisfying $0\leq\pi_j\leq 1$. A Markov jump linear system where the switching is governed by Markovian process is one example of stochastic jump linear systems. The switching rule under Markovian process obeys the structure $\pi(k+1) = \pi(k)P$, where $P$ is Markov transition probability matrix. This implies that the temporal evolution of a switching probability $\pi(k)$ depends on $P$ matrix. Unlike Markov jump systems, however, we do not impose any restriction on the stochastic switching rules at any time, and hence $\pi_j$ can form any arbitrary real fractional number between $0$ and $1$.

\begin{remark} (Stationary switching sequence)
A switching sequence for stochastic jump linear system is called stationary, if the occupation probability vector $\pi\left(k\right)$ remains stationary in time. In particular, a stationary deterministic switching sequence implies execution of a single mode (no switching). A stationary randomized switching sequence implies i.i.d. jump process.
\label{StationarySwitching}
\end{remark}

\section{Wasserstein Distance}
Here we assume that the initial state uncertainties are given by Gaussian PDF. Note that these types of initial state uncertainties, where the main sources usually come from measurement noise or sensor inaccuracy are common in real implementation level. Hence, two different types of uncertainties are involved in this problem; initial state uncertainties and stochastic switching laws. In order to measure the robustness against initial state uncertainties together with stochastic switching, we adopt Wasserstein distance which defines a metric on the manifold of PDFs.

\begin{definition} ({Wasserstein distance})
Consider the metric space $\ell_{2}\left(\mathbb{R}^{n}\right)$ and let the vectors $x_{1}, x_{2} \in \mathbb{R}^{n}$. Let $\mathcal{P}_{2}(\varsigma_{1},\varsigma_{2})$ denote the collection of all probability measures $\varsigma$ supported on the product space $\mathbb{R}^{2n}$, having finite second moment, with first marginal $\varsigma_{1}$ and second marginal $\varsigma_{2}$. Then the $L_{2}$ Wasserstein distance of order 2, denoted as $_{2}W_{2}$, between two probability measures $\varsigma_{1},\varsigma_{2}$, is defined as
\begin{align}\label{Wassdefn}
&_{2}W_{2}(\varsigma_{1},\varsigma_{2}) \triangleq \\ \nonumber & \left(\displaystyle\inf_{\varsigma\in\mathcal{P}_{2}(\varsigma_{1},\varsigma_{2})}\displaystyle\int_{\mathbb{R}^{2n}} \parallel x_{1}-x_{2}\parallel_{\ell_{2}\left(\mathbb{R}^{n}\right)}^{2} \: d\varsigma(x_{1},x_{2}) \right)
^{\frac{{1}}{2}}.
\label{W-dist}
\end{align}
\end{definition}
\begin{remark}
Intuitively, Wasserstein distance equals the \emph{least amount of work} needed to morph one distributional shape to the other, and can be interpreted as the cost for Monge-Kantorovich optimal transportation plan \cite{villani2003topics}. For notational ease, we henceforth denote $_{2}W_{2}$ as $W$. Further, one can prove (p. 208, \cite{villani2003topics}) that $W$ defines a metric on the manifold of PDFs.
\label{WassRemarkFirst}
\end{remark}
Next, we present new results for stability in terms of $W$.
 
\begin{proposition}\label{m.s.stable}
If we fix Dirac distribution as the reference measure, then distributional convergence in Wasserstein metric is \emph{necessary and sufficient} for convergence in m.s. sense.\label{WConvMeanSqConvProposition}
\end{proposition}
\begin{proof}
Consider a sequence of $n$-dimensional joint PDFs $\{\varsigma_{j}\left(x\right)\}_{j=1}^{\infty}$, that converges to $\delta\left(x\right)$ in distribution, i.e., $\displaystyle\lim_{j\rightarrow\infty} W\left(\varsigma_{j}(x), \delta(x)\right) = 0 = \displaystyle\lim_{j\rightarrow\infty} W^{2}\left(\varsigma_{j}(x), \delta(x)\right)$. From (\ref{Wassdefn}), we have
\begin{eqnarray}
&W^{2}\left(\varsigma_{j}(x), \delta(x)\right) = \displaystyle\inf_{\varsigma\in\mathcal{P}_{2}(\varsigma_{j}(x),\delta(x))} \mathbb{E}\left[\parallel X_{j} - 0 \parallel_{\ell_{2}\left(\mathbb{R}^{n}\right)}^{2}\right]\label{m.s.stab}\\ \nonumber 
& = \mathbb{E}\left[\parallel X_{j} \parallel_{\ell_{2}\left(\mathbb{R}^{n}\right)}^{2}\right].
\end{eqnarray}
where the random variable $X_{j} \sim \varsigma_{j}\left(x\right)$, and the last equality follows from the fact that $\mathcal{P}_{2}(\varsigma_{j}(x),\delta(x)) = \{\varsigma_{j}(x)\}$ $\forall \: j$, thus obviating the infimum. From \eqref{m.s.stab}, $\displaystyle\lim_{j\rightarrow\infty} W\left(\varsigma_{j}(x), \delta(x)\right) = 0 \Rightarrow \displaystyle\lim_{j\rightarrow\infty} \mathbb{E}\left[\parallel X_{j} \parallel_{\ell_{2}}^{2}\right] = 0$, establishing distributional convergence to $\delta(x) \Rightarrow$ m.s. convergence. Conversely, m.s. convergence $\Rightarrow$ distributional convergence, is well-known \cite{grimmett2001probability} and unlike the other direction, holds for arbitrary reference measure.
\end{proof}

\begin{proposition}($W^{2}$ between Gaussian and Dirac PDF (see e.g., p. 160-161,\cite{hassani2013mathematical}))
The Wasserstein distance between Gaussian and Dirac PDF supported on $\mathbb{R}^{n}$, with respective joint PDFs $\varsigma = \mathcal{N}\left(\mu,\Sigma\right)$ and $\delta\left(x\right) = \displaystyle\lim_{\mu,\Sigma \rightarrow 0} \mathcal{N}\left(\mu,\Sigma\right)$, is given by,
\begin{eqnarray}\label{staticW}
W^{2}\left(\mathcal{N}\left(\mu,\Sigma\right), \delta\left(x\right)\right) = \parallel \mu \parallel_{\ell_{2}\left(\mathbb{R}^{n}\right)}^{2} + \: \text{tr}\left(\Sigma\right).
\label{GaussianDiracW}
\end{eqnarray}
\label{GaussToDiracCorollary}
\end{proposition}

Note that the reference that is the stationary equilibrium point in most cases is assumed to be deterministic with Dirac distribution. Therefore, $W$ provides a quantitative measurement, representing how the state PDF converges to Dirac PDF, where the state reference is placed. Having provided the m.s. stability with $W$, the robustness of the stochastic jump linear system with respect to initial state uncertainties is studied in the next section.

\section{Robustness Analysis}
The robustness analysis has two stages; uncertainty propagation and $W$ computation. With given initial state PDF we propagate the initial PDF along dynamics, followed by the robustness analysis of stochastic jump linear systems in terms of $W$ distance.

\subsection{Uncertainty Propagation}
Note that two different types of uncertainties are involved in the problem. First, initial state uncertainties are given as Gaussian PDF $\mathcal{N}(\mu_0,\Sigma_0)$. Second, stochastic switching laws $\pi=[\pi_1,\pi_2, \cdots,\pi_m]$ result in a randomness of system trajectories. Reflecting these two types of uncertainties, we need to propagate the state PDF properly.

Due to stochastic switching laws the state PDF $\rho(k)$ does not remain Gaussian, but rather form MoG even though the initial state PDF starts from a Gaussian PDF $\mathcal{N}(\mu_0,\Sigma_0)$. The closed form expression of the state PDF can be represented by the following lemma.

\begin{proposition}\label{Prop:4.1}
Given $m$ absolutely continuous random variables $X_{1}, \hdots, X_{m}$, with respective cumulative distribution function(CDF) $F_{j}\left(x\right)$, and PDF $\varsigma_{j}\left(x\right)$, $\forall j \in \mathcal{I}$. Let $X \triangleq X_{j}$, with probability $\alpha_{j} \in [0,1]$, $\displaystyle\sum_{i=1}^{m} \alpha_{j} = 1$. Then, the CDF and PDF of $X$ are given by 
\begin{align}
F\left(x\right) = \displaystyle\sum_{j=1}^{m} \alpha_{j} F_{j}\left(x\right),\quad\varsigma\left(x\right) = \displaystyle\sum_{j=1}^{m} \alpha_{j} \varsigma_{j}\left(x\right).
\end{align}
\label{RandomVarProposition}
\end{proposition}
\begin{proof}
\begin{align*}
F\left(x\right) &\triangleq \mathbb{P}\left(X \leq x\right) = \displaystyle\sum_{j=1}^{m} \mathbb{P}\left(X=X_{j}\right) \mathbb{P}\left(X_{j} \leq x\right)\\
 &= \displaystyle\sum_{j=1}^{m} \alpha_{j} F_{j}\left(x\right).
\end{align*}
where we have used the law of total probability. Since each $X_{j}$ and hence $X$, is absolutely continuous, we have $\varsigma\left(x\right) = \displaystyle\sum_{j=1}^{m} \alpha_{j} \varsigma_{j}\left(x\right)$.
\end{proof}

\begin{corollary}\label{corollary:mog_closedform}
Consider a stochastic jump linear sytem $\left(\pi\left(k\right), \{A_{j}\}_{j=1}^{m}, \mathcal{I}\right)$ with initial PDF $\rho_{0} = \displaystyle \mathcal{N}\left(\mu_{{0}},\Sigma_{{0}}\right)$. Then the state PDF at time $k$, denoted by $\rho\left(k\right)$, is given by
\begin{align}
\rho\left(k\right) &= \displaystyle\sum_{j_{k}=1}^{m}\displaystyle\sum_{j_{k-1}=1}^{m} \hdots \displaystyle\sum_{j_{1}=1}^{m} \displaystyle\left(\prod_{r=1}^{k} \pi_{j_{r}}\left(r\right)\right)\qquad  \nonumber\\ 
&\left.\qquad\qquad\qquad\qquad\mathcal{N}\left(A_{j_k}^*\mu_{{0}}, A_{j_k}^*\Sigma_{{0}}A_{j_k}^{*{\top}}\right)\right..
\label{dtSJLSstatePDF}
\end{align}
where $\displaystyle A_{j_{k}}^*\triangleq \prod_{r=k}^{1}A_{j_r}=A_{j_k}A_{j_{k-1}}\hdots A_{j_{2}}A_{j_{1}}$.
\end{corollary}
\begin{proof}
Starting from $\rho_{0}$ at $k=0$, the modal PDF at time $k=1$, is given by
\begin{align}
\rho_{j}(1) &= \displaystyle\mathcal{N}\left(A_{j}\mu_{{0}}, A_{j}\Sigma_{{0}}A_{j}^{\top}\right), \: j=1,\cdots ,m \label{ModalPDFatTime1}
\end{align}
which follows from the fact that linear transformation of an MoG is an equal component MoG with linearly transformed component means and congruently transformed component covariances (see Theorem 6 and Corollary 7 in \cite{ali2011convergence}). From Proposition \ref{RandomVarProposition}, it follows that the state PDF at $k=1$, is
\begin{eqnarray}
\rho(1) = \displaystyle\sum_{j_{1}=1}^{m}\pi_{j_{1}}(1)\mathcal{N}\left(A_{j_{1}}\mu_{{0}}, A_{j_{1}}\Sigma_{{0}}A_{j_{1}}^{\top}\right),
\label{dtSJSPDFatTime1}
\end{eqnarray}
where $\pi_{j_{1}}(1)$ is the occupation probability for mode $j_{1}$ at time $k=1$. Notice that (\ref{dtSJSPDFatTime1}) is an MoG with $m$ component Gaussians. Proceeding likewise from this $\rho(1)$, we obtain
\begin{eqnarray}
\rho_{j}(2) &=& \displaystyle\sum_{j_{1}=1}^{m}\pi_{j_{1}}(1)\mathcal{N}\big((A_{j}A_{j_{1}})\mu_{{0}},\nonumber\\
&& (A_{j}A_{j_{1}})\Sigma_{{0}}(A_{j}A_{j_{1}})^{\top}\big), \quad j=1,\hdots,m,\\
\rho(2)&=&\displaystyle\sum_{j_{2}=1}^{m}\displaystyle\sum_{j_{1}=1}^{m}\pi_{j_{2}}(2)\pi_{j_{1}}(1)\mathcal{N}\big((A_{j_{2}}A_{j_{1}})\mu_{{0}},\nonumber\\
&& (A_{j_{2}}A_{j_{1}})\Sigma_{{0}}(A_{j_{2}}A_{j_{1}})^{\top}\big).
\label{ModalStatePDFatTime2}
\end{eqnarray}
Continuing with this recursion till time $k$, we arrive at (\ref{dtSJLSstatePDF}), which is an MoG with $m^{k}$ components. 
\end{proof}

According to Corollary \ref{corollary:mog_closedform}, the MoG has total $m^k$ components of Gaussian PDF at time $k$ as described in Fig.\ref{fig:mog} ($m=2$). Therefore, the growth of components in MoG is exponential in time. Consequently, the computation for uncertainty propagation will blow up even for a finite time and finite number of mode. Hence, for the practical purpose we need to find an alternative way to replace the MoG. 

\subsection{Wasserstein Computation}
In order to cope with the exponential growth of Gaussian components in MoG form as explained previously, we need to replace it. In this subsection, we introduce a new method to substitute the MoG PDF into an equivalent PDF in the $W$ space.

Following lemma presents that we can compute mean-covariance pairs $\left(\widehat{\mu},\widehat{\Sigma}\right)$ for any mixture PDF.
\begin{figure*}
\centering
\includegraphics[width=0.9\textwidth,height=10cm]{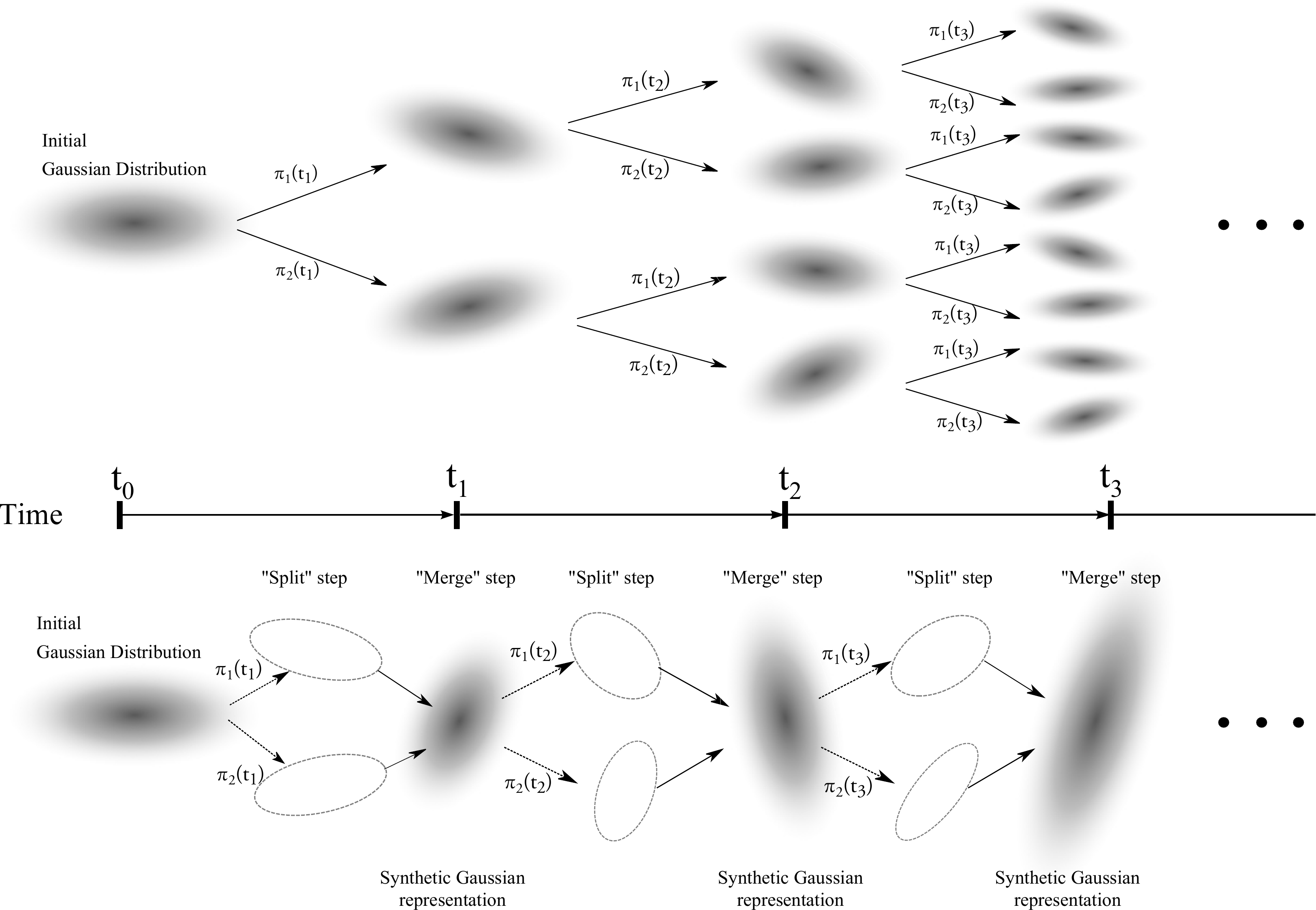}
\caption{Schematic of Uncertainty Propagation for a given Initial Gaussian PDF. Top figure shows spatio-temporal evolution of the state PDF in Mixture of Gaussian form that has $2^k$ Gaussian components at time $k$. Bottom figure presents the``Split-and-Merge'' algorithm for synthetic Gaussian representation, where the number of Gaussian components remains constant and is maximally $2$ in this case.}
\label{fig:mog}
\end{figure*}
\begin{lemma} \label{lemma:4.2}
Consider any mixture PDF $\varsigma(x) = \displaystyle\sum_{j=1}^{m} \pi_{j} \varsigma_{j}(x)$, with component mean-covariance pairs $\left(\mu_{j},\Sigma_{j}\right)$, $j=1,\hdots,m$. Then, the mean-covariance pair $\left(\widehat{\mu},\widehat{\Sigma}\right)$ for the mixture PDF $\varsigma(x)$, is given by
\begin{align}
\widehat{\mu} = \sum_{j=1}^{m}\pi_{j} \mu_{j},\quad\widehat{\Sigma} = \sum_{j=1}^{m} \pi_{j}\left(\Sigma_{j} + \left(\mu_{j}-\widehat{\mu}\right)\left(\mu_{j}-\widehat{\mu}\right)^{\top}\right)
\label{MeanCovHat}
\end{align}
\label{MeanCovMixPDF}
\end{lemma}
\begin{proof}
The mean of the mixure PDF is
\begin{align*}
\widehat{\mu} &\triangleq \displaystyle\int_{\mathbb{R}^{n}} x \varsigma(x) dx = \displaystyle\int_{\mathbb{R}^{n}} x \displaystyle\sum_{j=1}^{m} \pi_{j} \varsigma_{j}(x) dx \\
& = \displaystyle\sum_{j=1}^{m} \pi_{j} \displaystyle\int_{\mathbb{R}^{n}} x \varsigma_{j}(x) dx  = \displaystyle\sum_{j=1}^{m}\pi_{j} \mu_{j}
\end{align*}
Also, the covariance of the mixture PDF is
\begin{align*}
\widehat{\Sigma} &\triangleq \mathbb{E}\left[\left(x-\widehat{\mu}\right) \left(x-\widehat{\mu}\right)^{\top}\right] = \mathbb{E}\left[x x^{\top}\right] - \widehat{\mu}\widehat{\mu}^{\top} \\
&= \displaystyle\int_{\mathbb{R}^{n}} x x^{\top} \displaystyle\sum_{j=1}^{m} \pi_{j} \varsigma_{j}(x) dx - \widehat{\mu}\widehat{\mu}^{\top}\\
& = \displaystyle\sum_{j=1}^{m} \pi_{j} \displaystyle\int_{\mathbb{R}^{n}} \left(x - \widehat{\mu} + \widehat{\mu}\right) \left(x - \widehat{\mu} + \widehat{\mu}\right)^{\top} \varsigma_{j}\left(x\right) dx - \widehat{\mu}\widehat{\mu}^{\top} \\
&=\displaystyle\sum_{j=1}^{m}\pi_{j}\left(\Sigma_{j}+\left(\mu_{j}-\widehat{\mu}\right)\left(\mu_{j}-\widehat{\mu}\right)^{\top}\right)
\end{align*}
\end{proof}

By using Lemma \ref{lemma:4.2}, we can construct a synthetic Gaussian that has a mean $\widehat{\mu}$ and a covariance $\widehat{\Sigma}$. Since a real state PDF with an MoG form may hold higher moments other than first and second, this synthetic Gaussian that has only first and second moment does not preserve the exact information. However, most importantly, the following lemma and theorem show that ``$W$'' distance preserves exactly identical information between MoG PDF $\rho(k)$ and a synthetic Gaussian PDF $\mathcal{N}(\widehat{\mu},\widehat{\Sigma})$.

\begin{lemma}\label{lemma:4.3}
At any time $k$, let the state PDF for stochastic jump linear system $\rho(k)$ be of the form in \eqref{dtSJLSstatePDF}. If we define $W(k)\triangleq W\left(\rho(k,x),\delta(x)\right)$ and $W_j(k)\triangleq W\left(\mathcal{N}_j(k,x),\delta(x)\right)$, then we have
\begin{align}
W^2(k) = \sum_{j=1}^{m}\pi_j(k)W_j^2(k),\quad \forall k
\end{align}
\end{lemma}
\begin{proof}
From (\ref{Wassdefn}) and Proposition \ref{Prop:4.1}, we have
\begin{align}
\nonumber &W^{2} = \displaystyle\int_{\mathbb{R}^{n}} \parallel x \parallel_{\ell_{2}\left(\mathbb{R}^{n}\right)}^{2} \varsigma(x) dx\\
&= \displaystyle\int_{\mathbb{R}^{n}} \parallel x \parallel_{\ell_{2}\left(\mathbb{R}^{n}\right)}^{2} \displaystyle\sum_{j=1}^{m} \pi_{j} \varsigma_{j}(x) dx \nonumber \\
& = \displaystyle\sum_{j=1}^{m} \pi_{j} \displaystyle\int_{\mathbb{R}^{n}} \parallel x \parallel_{\ell_{2}\left(\mathbb{R}^{n}\right)}^{2} \varsigma_{j}(x) dx \nonumber = \displaystyle\sum_{j=1}^{m} \pi_{j} W_{j}^{2} \nonumber \\ 
&\Rightarrow W^2(k) = \displaystyle\sum_{j=1}^{m} \pi_{j}(k) W_{j}^{2}(k),\:\forall k\label{eqn:W^2=sum_piW^2} 
\end{align}
\end{proof}

Above lemma shows that the $W$ distance between $\rho(k)$ and $\delta(k)$ is equivalent to the $\pi(k)$ weighted sum of componentwise $W$ between individual Gaussian and $\delta(k)$. Now we further prove that we can avoid this componentwise convex combination expression\eqref{eqn:W^2=sum_piW^2} by a following theorem.
\begin{theorem}\label{theorem:4.1}
At any given time $k$, let the state PDF for stochastic jump linear system $\rho(k)$ be of the form \eqref{dtSJLSstatePDF}. In addition, let the instantaneous mean and covariance of mixture PDF $\rho(x,k)$ be $
\widehat{\mu}$ and $\widehat{\Sigma}$, respectively. If we denote $\widehat{W}(k)\triangleq W\left(\mathcal{N}\left(\widehat{\mu}(k),\widehat{\Sigma}(k)\right),\delta(x)\right)$ and $W(k)\triangleq W\left(\rho(k,x),\delta(x)\right)$, then we have
\begin{equation}
\qquad\widehat{W}(k) = W(k),\quad\forall k
\end{equation}
\end{theorem}
\begin{proof}
From Proposition \ref{GaussToDiracCorollary}, we have
\begin{align}
&\widehat{W}^{2} = \parallel \widehat{\mu}\parallel_{\ell_{2}\left(\mathbb{R}^{n}\right)}^{2} + \text{tr}(\widehat{\Sigma}) \overset{(\ref{MeanCovHat})}{=} \nonumber \\ 
&\widehat{\mu}^{\top}\widehat{\mu} + \:\text{tr}\left(\displaystyle\sum_{j=1}^{m}\pi_j(\Sigma_{j} + (\mu_{j}-\widehat{\mu})(\mu_{j}-\widehat{\mu})^{\top}\right),
\label{dtSJLSThmProof}
\end{align}
Since$\tr(\cdot)$ is linear operator and $\displaystyle\sum_{j=1}^{m}\pi_j = 1$, we can simplify \eqref{dtSJLSThmProof} as 
\begin{align}
&\widehat{W}^{2} = \widehat{\mu}^{\top}\widehat{\mu}
+ \: \displaystyle\sum_{j=1}^{m}\pi_{j}\text{tr}\left(\Sigma_{j}\right) + \: \text{tr}\left(\displaystyle\sum_{j=1}^{m} \pi_{j}\mu_{j}\mu_{j}^{\top}\right) - \nonumber \\ & \text{tr}\left(\left(\displaystyle\sum_{j=1}^{m} \pi_{j}\mu_{j}\right)\widehat{\mu}^{\top}\right) - \text{tr}\left(\widehat{\mu}\left(\displaystyle\sum_{j=1}^{m} \pi_{j}\mu_{j}\right)^{\top}\right) + \nonumber \\ &\text{tr}\left(\widehat{\mu}\widehat{\mu}^{\top}\right).
\label{maineq1}
\end{align}

Now, we recall from (\ref{MeanCovHat}) that $\widehat{\mu} = \displaystyle\sum_{j=1}^{m} \pi_{j} \mu_{j}$, and that $\widehat{\mu}^{\top}\widehat{\mu} = \:\text{tr}\left(\widehat{\mu}^{\top}\widehat{\mu}\right) = \:\text{tr}\left(\widehat{\mu}\widehat{\mu}^{\top}\right)$. 

Consequently, the first, fourth, fifth and sixth term in the right-hand-side of (\ref{maineq1}) cancel out, resulting
\begin{align*}
& \widehat{W}^{2} = \displaystyle\sum_{j=1}^{m}\pi_{j}\text{tr}\left(\Sigma_{j}\right) + \: \displaystyle\sum_{j=1}^{m} \pi_{j} \: \text{tr}\left(\mu_{j}\mu_{j}^{\top}\right) \\ &= \displaystyle\sum_{j=1}^{m}\pi_{j} \left(\parallel \mu_{j} \parallel_{\ell_{2}\left(\mathbb{R}^{n}\right)}^{2} + \: \text{tr}\left(\Sigma_{j}\right)\right) = \displaystyle\sum_{j=1}^{m}\pi_{j} W_{j}^{2}
\end{align*}
Finally, from Lemma \ref{lemma:4.3}, we obtain
\begin{eqnarray*}
\widehat{W}^2(k) = \sum_{j=1}^{m}\pi_j(k)W_j^2(k) = W^2(k)\nonumber\\
\Rightarrow \widehat{W}(k) = W(k),\quad \forall k\qquad
\end{eqnarray*}
\end{proof}

Theorem \ref{theorem:4.1} states that at any time $k$, there always exists a Gaussian PDF $\mathcal{N}\left(\widehat{\mu},\widehat{\Sigma}\right)$ such that the distance between $\rho(x)$ and $\delta(x)$ is equivalent to the distance between $\mathcal{N}(\widehat{\mu},\widehat{\Sigma})$ and $\delta(x)$ as shown in Fig. \ref{IllustrateMOGtheorem}. Further, at each time, such a Gaussian PDF can be constructed (using (\ref{MeanCovHat})) from component Gaussians available in closed form. The practical utility of Theorem \ref{theorem:4.1} stems from the fact that it obviates the need to compute the MoG PDF $\rho(k)$ for performance analysis, which incurs exponential growth in computational complexity, as discussed before. On the contrary, the computational complexity for $\widehat{W}(k)$ remains constant with time, and admits a closed form solution. This can be leveraged via the ``split-and-merge" algorithm illustrated in Fig. \ref{fig:mog}. Starting with an initial Gaussian PDF, linear modal dynamics results in $m$ modal Gaussian PDFs (``split step"). Instead of computing the MoG state PDF, one would then construct a synthetic Gaussian $\mathcal{N}\left(\widehat{\mu},\widehat{\Sigma}\right)$ (``merge step") followed by $\widehat{W}$ computation in closed form, and repeat thereafter. Thus, at any time $k$, we only have $m$ mean vectors and covariance matrices to work with.

\begin{figure}
\begin{center}
\includegraphics[width=0.49\textwidth]{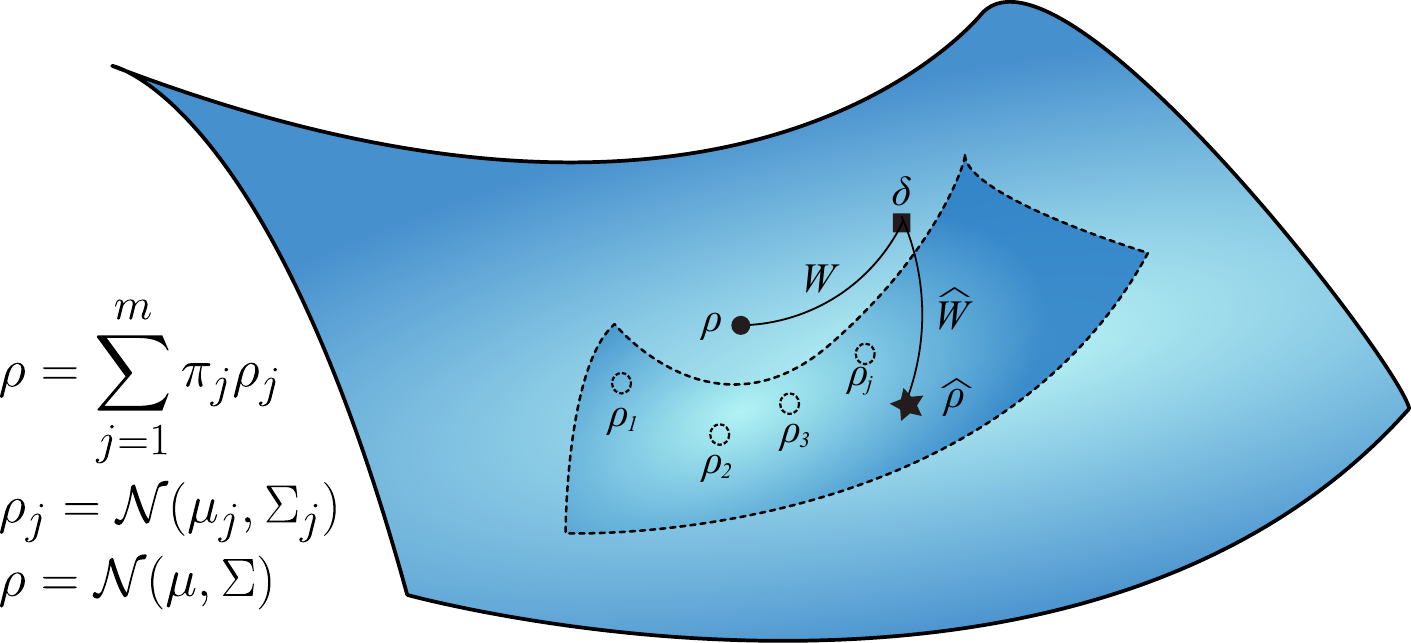}
\caption{Illustration of Theorem \ref{theorem:4.1}, showing that given MoG PDF $\rho$, we can construct Gaussian $\widehat{\rho}$ such that $W = \widehat{W}$, where $W\triangleq W\left(\rho,\delta\right)$, and $\widehat{W}\triangleq W\left(\widehat{\rho},\delta\right)$.}
\label{IllustrateMOGtheorem}
\end{center}
\end{figure}

\section{Academic Example}
This example demonstrates how the proposed method for probabilistic robustness analysis can be used for stochastic jump linear systems. Note that the switching law can be any random process, but we consider that $\pi(k)$ is governed by a Markov process. The discrete-time system dynamics for the Markov jump linear system is given by
\begin{align*}
x(k+1) = A_{\sigma}x(k), \qquad \sigma\in \mathbb{Z}=\{1,2\},
\end{align*}
\begin{align*}
A_{1} = \begin{bmatrix}
		  0.7 & 0 \\
		  0 & 1\\
		  \end{bmatrix},
\qquad
A_{2} = \begin{bmatrix}
		  1 & 0 \\
		  0 & 0.85\\
		  \end{bmatrix},
\end{align*}
where the initial switching probability $\pi(0)$ and the Markov transition probability matrix $P$ are as follows.
\begin{align*}
\pi_{0} = [0.5, 0.5], \quad
P = \begin{bmatrix}
0.75 & 0.25\\
0.2 & 0.8
\end{bmatrix}.
\end{align*}
We assume that the initial condition has uncertainty and it is described as a following Gaussian PDF.
\begin{align*}
\mathcal{N}(\mu_{0},\Sigma_{0}), \text{ with }\mu_{0} = [5,5]^{\top}, \text{ and }\Sigma_{0} = \begin{bmatrix}
		  0.1 & 0 \\
		  0 & 0.1\\
		  \end{bmatrix}.
\end{align*}

Firstly, we compute the propagation of the PDF along each dynamics without switching by the linear recursions that are $\mu(k+1)=A_j\mu(k)$ and $\Sigma(k+1)=A_j\Sigma(k)A_j^{\top}$, $\forall j=1,2$.
The propagation of the state PDF corresponding to each mode without switching is shown in Fig. \ref{fig_sample_propagation}(a). As expected from the definition of the dynamics, the modal PDFs without switching move along a particular axis. In contrast, Fig. \ref{fig_sample_propagation}(b) shows the evolution of the density function under Markov switching, which is computed from  Corollary \ref{corollary:mog_closedform} in MoG form with  $\pi_{0}$ and transition probability matrix $P$. It can be shown that the Markov jump linear system considered here is stable and the PDF will converge to $\delta(x)$. This is supported by the convergence of  $W(k)$ to zero as shown in Fig. \ref{fig_ex1_wasserstein}. In this $W$ plot, each mode dynamics does not converge and shows steady-state error while the Markovian switching dynamics asymptotically converges to zero. For a given uncertain bound of initial state the proposed method confirms the system robustness with respect to both initial state uncertainty and Markovian switching.

Without using techniques introduced in section IV, it is practically impossible to propagate density functions and calculate $W$ even for a Markov jump linear system with two modes. The number of Gaussian components that represents the state PDF after $N$ time steps is $2^{N}$, which soon becomes computationally intractable. For $m$ modes stochastic jump linear system, the growth rate is $m^N$. The proposed method in this paper, however, has maximum $m$ number of Gaussian components regardless of time evolution.

\begin{figure}[!ht] 
\begin{center}
\subfigure[Individual system PDF propagation]{\includegraphics[width=0.4\textwidth]{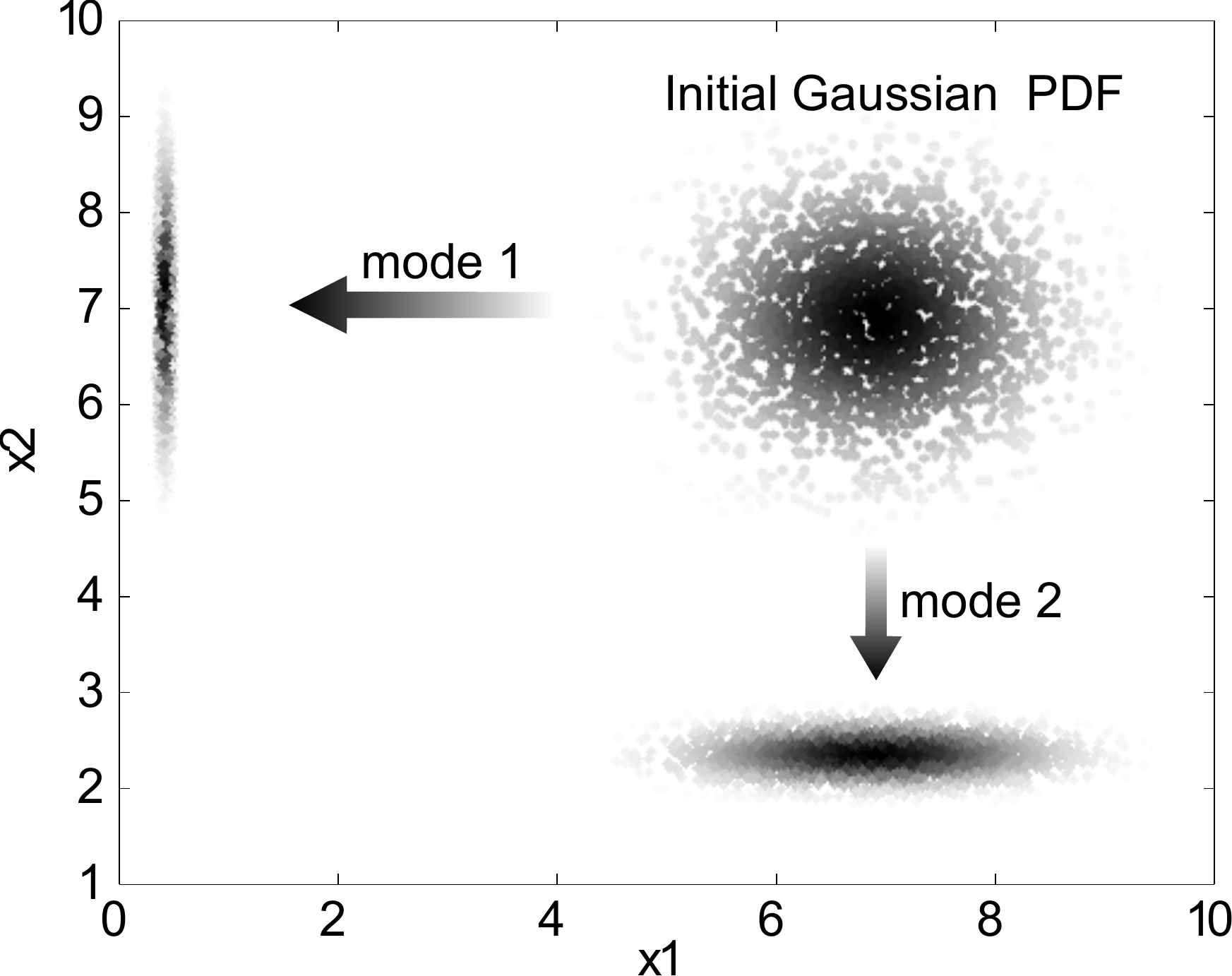}}
\subfigure[PDF propagation under switching]{\includegraphics[width=0.4\textwidth]{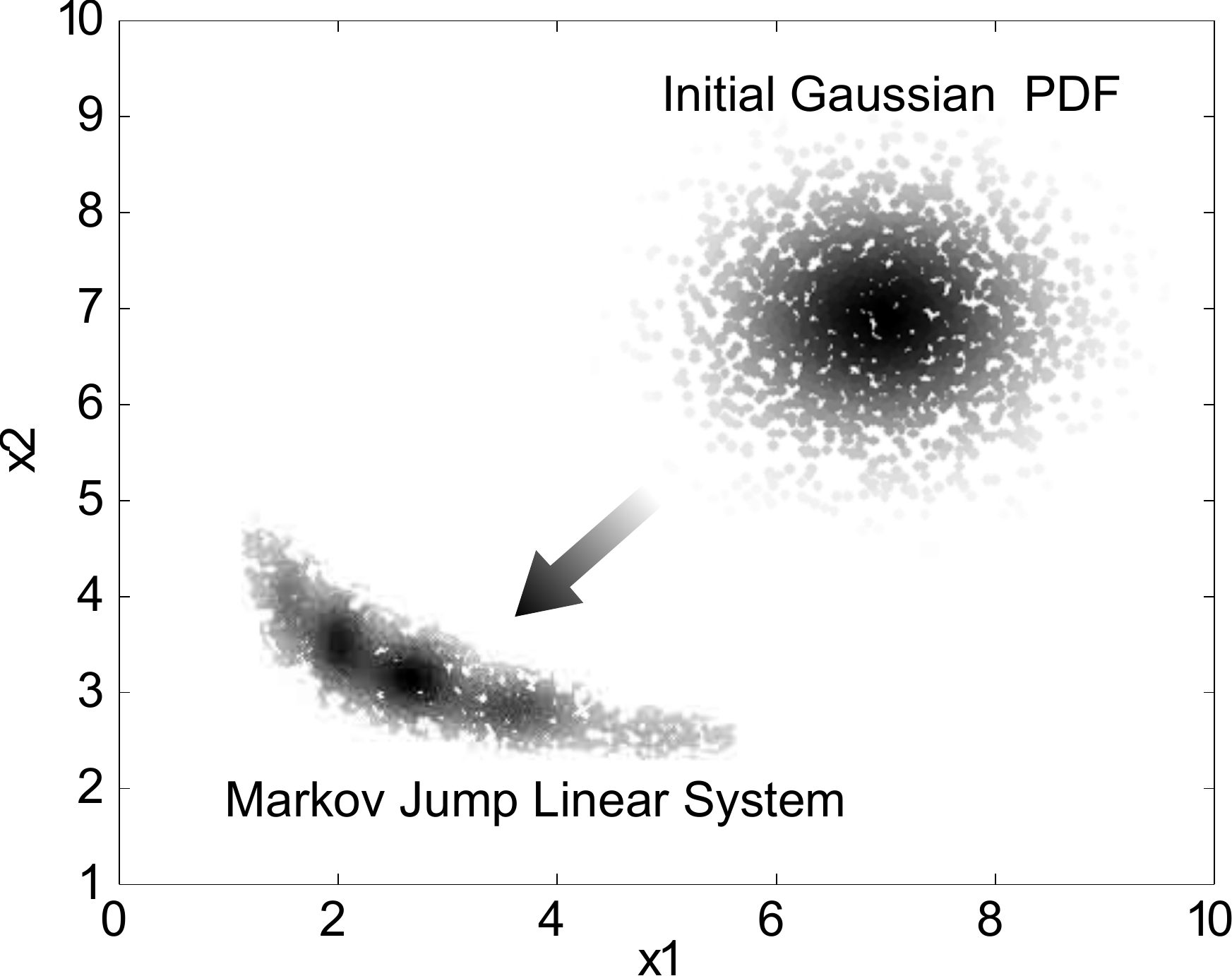}}
\subfigure[Corresponding $W(k)$]{\includegraphics[width=0.4\textwidth]{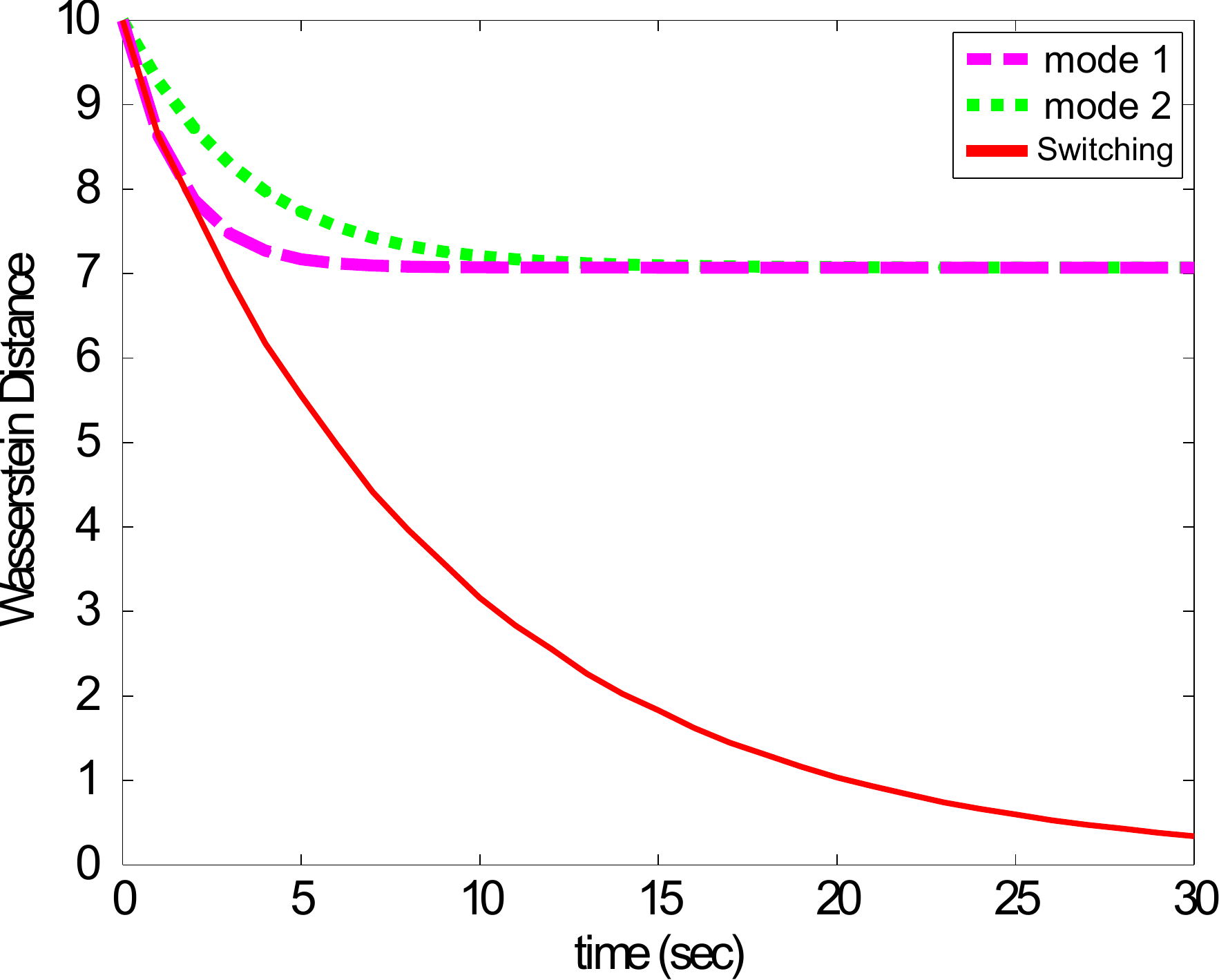}\label{fig_ex1_wasserstein}}
\caption{Simulation result of PDF Propagation and Robustness Analysis for Markov Jump Linear System with given initial Gaussian state PDF}\label{fig_sample_propagation}
\end{center}
\end{figure}

\section{Conclusions}
In this paper, we proposed a new method for robustness analysis of the stochastic jump linear systems with given initial state uncertainties. In general, the system state contains uncertainties which come from measurement noise or sensor inaccuracy. We assumed that theses initial state uncertainties are represented by Gaussian PDF. Starting with Gaussian PDF, the state PDF forms MoG where the number of Gaussian components grows exponentially over time. To cope with the computational complexity caused by the structure in MoG, we proved that there always exists a synthetic Gaussian $\mathcal{N}(\widehat{\mu},\widehat{\Sigma})$ which has an equidistance with state PDF $\rho(x)$ in terms of $W$. In the example, we showed the efficiency of the proposed method.

\bibliographystyle{plain}
\bibliography{Analysis_ACC2014}
\end{document}